\documentclass{article}
\usepackage{graphicx} 
\usepackage{geometry}
\usepackage[utf8]{inputenc}
\usepackage[T1]{fontenc}
\usepackage{setspace}
\usepackage{amsmath}
\usepackage{graphicx}
\usepackage{float}
\usepackage{microtype}
\usepackage{amssymb}
\usepackage{amsthm}
\usepackage{physics}

\newtheorem{mytheorem}{Theorem}

\newtheorem{mycor}{Corollary}

\usepackage{tikz}
\usetikzlibrary{shapes.geometric, arrows}
\usepackage{hyperref}
\usepackage{pgf-pie}
\usepackage{graphicx}
\usepackage{pdfpages}
\usepackage{color}
\usepackage{subcaption}
\usepackage[]{algorithm2e}
\usepackage{bbold}
\usepackage{mathtools}
\usepackage{enumerate}
\usepackage{tcolorbox}
\usepackage{adjustbox}
\usepackage{floatflt}
\usepackage{comment}
\usepackage{svg}
\usepackage[OT1]{fontenc}
\usepackage{hyperref}
\usepackage[
    backend=biber,
    style=nature]{biblatex} 
\usepackage{authblk}

\newcommand{\M}{\mathcal{M}}
\newcommand{\sumkl}{\sum_{k,l = 0}^{d-1}}
\newcommand{\sumnm}{\sum_{n,m = 0}^{d-1}}
\newcommand{\Pkl}{P_{k,l}}
\newcommand{\Okl}{\Omega_{k,l}}

\newcommand{\ckl}{c_{k,l}}
\newcommand{\Tij}{T_{i,j}}
\newcommand{\WWij}{W_{i,j} \otimes W_{i,j}^\ast}

\title{Special features of the Weyl-Heisenberg Bell basis imply unusual entanglement structure of Bell-diagonal states}
\author[1,*]{Christopher Popp}
\author[1,**]{Beatrix C. Hiesmayr}
\affil[1]{University of Vienna, Faculty of Physics, Währingerstrasse 17, 1090 Vienna, Austria}
\affil[*]{christopher.popp@univie.ac.at}
\affil[**]{Beatrix.Hiesmayr@univie.ac.at}

\date{September 2023}

\begin{document}

\maketitle
\begin{abstract}
Bell states are of crucial importance for entanglement based methods in quantum information science. Typically, a standard construction of a complete orthonormal Bell-basis by Weyl-Heisenberg operators is considered. We show that the group structure of these operators has strong implication on error correction schemes and on the entanglement structure within Bell-diagonal states. In particular, it implies an equivalence between a Pauli channel and a twirl channel. Interestingly, other complete orthonormal Bell-bases do break the equivalence and lead to a completely different entanglement structure, for instance in the share of PPT-entangled states. In detail, we find that the standard Bell basis has the highest observed share on PPT-states and PPT-entangled states compared to other Bell bases. In summary, our findings show that the standard Bell basis construction exploits a very special structure with strong implications to quantum information theoretic protocols if a deviation is considered. 

\end{abstract}
\section{Introduction}
Leveraging quantum phenomena in technology offers new resources that enable methods with better performance than classically limited methods in the field of information theory and related applications like communication, computation, simulation, metrology or cryptography~~\cite{distQuantComp,quantSimReview, QuantCryptBell,oneStepQSDC, twoStepQuantDirCom}. 

Entanglement is one of the main resources for novel methods for quantum information processing like super-dense coding~\cite{superdenseCodingBennet} or teleportation~\cite{teleportingWeyl}, as well as in other fields like e.g. medical sciences~\cite{jpet1, HiesmayrMoskalGenuine, HiesmayrMoskalWitness, jpet2,jpet3}. Despite its relevance for our quantum-theoretic understanding of nature~\cite{epr, bellEpr, loopholeFreeBell}  and for practical applications in quantum technology, this phenomenon is far from being well understood. 
Two entangled qubits are the simplest system to observe entanglement. The qubit Bell states~\cite{bellStates} are a set of maximally entangled states, which form a basis of this bipartite Hilbert space. A bipartite pure state is called ``maximally entangled'', if the reduced state for each subsystem is the maximally mixed state, so all information about the state lies in the correlations between the subsystems. Most applications leveraging entanglement as a resource use these special states. The amount of entanglement two parties possess if they share a maximally entangled qubit Bell state is named ``ebit''~\cite{wilde_2017}. Recently, multi-level quantum systems called ``qudits'' (with the ``$d$'' indicating the general dimension of the system) have been shown to offer potential advantages for applications~\cite{highDimQCReview, qditsQCWang}. The notion of Bell states can be generalized to the qudit systems~\cite{Sych}, in which case two parties holding a  maximally entangled bipartite qudit state are said to possess one ``edit''.

A frequently used standard construction for a basis of bipartite Bell states of arbitrary dimension is based on a set of operators called ``Weyl-Heisenberg'' operators, which can be seen as generalization of the two-dimensional Pauli matrices to other dimensions~\cite{teleportingWeyl}. Applying these operators locally to the standard maximally entangled state, a basis of Bell states can be generated~\cite{baumgartner1}. Due to their properties, this specific ``standard'' Bell basis is often used in applications, but as we will show in this contribution, other Bell bases exist, which differ in some relevant properties. 

For real applications, decoherence due to interaction with the environment generally disturbs a pure state and the resource in form of an edit is destroyed or cannot be used effectively. For this reason, quantum error correction~\cite{bennetMixedStateEntPurErrCorr, cssErrCorr} and entanglement purification/distillation~\cite{bennetEntPuri} are required. While error correction transforms a certain disturbed state back to some logical pure state, entanglement purification processes several weakly entangled or disturbed states to produce fewer but stronger or maximally entangled states. Both concepts have been generalized to qudit systems. The Weyl-Heisenberg operators thereby play a crucial role in the form of defining a ``nice'' error basis~\cite{knillQuditStabilizer, nonbinaryErrBases, nadkarniQuditStabilizer} or to use the corresponding Bell states as target states~\cite{horoRedCritAndDist, GernotEffEntPuriQudits, vollbrechtEffDistBeyondQubits}.
One phenomenon, which does not appear for bipartite qubits but for dimension of the subsystems $d\geq3$ is bound entanglement~\cite{distillation,Halder_BE, lockhart_e, Bruss_BE, Slater2019JaggedIO,hiesmayrLoeffler}. Bound entangled states are entangled states, but they cannot be used for entanglement purification. It is known that all entangled states with positive partial transposition (PPT) are bound entangled, while the existence of bound entangled states with negative partial transposition (NPT) is still an open problem. In general, the separability problem to decide whether a given PPT quantum state is entangled or separable is an NP-hard problem~\cite{nphard, nphard-strong}. 

Bipartite quantum states that are diagonal in the standard Bell basis are called ``Bell-diagonal''. Bell-diagonal states arise naturally in noise environments affecting Bell states or in several applications like entanglement distillation or quantum key distribution \cite{QuantCryptBell, qutritQKD}. If they are constructed via the Weyl-Heisenberg operators, special algebraic and geometric properties of the states can be leveraged for some insight concerning this complex structure of entanglement~\cite{baumgartner1, hiesmayrLoeffler, baumgartnerHiesmayr}. Recently, the authors combined analytical and numerical methods to investigate the systems of bipartite qutrits ($d=3$) and ququarts ($d=4$) in detail~\cite{PoppACS, PoppQutritsAndQuquarts, PoppJoss2023, hiesmayr1}. As a result, $95\%$/$77\%$ of all PPT (standard) Bell-diagonal qutrits/ququarts can be classified as PPT-entangled or separable. Moreover, it was shown that the group structure in the set of the standard Weyl-Heisenberg-constructed Bell states is highly relevant for the entanglement structure.

In this work, we demonstrate properties of the standard Bell basis constructed via the Weyl-Heisenberg operators and show that this basis has special properties among a set of generalized Bell bases with significant implications for the corresponding systems of Bell-diagonal states.
The paper is organized as follows:
First, we define a set of operators called ``Weyl-Twirl'' operators, for which all elements of the standard Bell basis are eigenstates and demonstrate an equivalence between the map to the set of Bell-diagonal states and the randomized application of those operators. Then, we present some applications of this ``Weyl-Twirl'' with relevance for the separability problem of Bell-diagonal states and a simple error correction scheme for Bell states. 
Second, by generalizing the construction based on the Weyl-Heisenberg operators, we present a family of Bell bases, which contains the standard Bell basis, but are generally not unitarily equivalent. We show that, even though the generalized Bell bases also consist of maximally entangled orthonormal Bell states, several characteristics and properties of the standard Bell basis relevant for applications do not exist in the generalized case. These differences significantly affect the entanglement structure of Bell-diagonal states and the detection of PPT-entangled states. Finally, we conclude with a summary of our findings.

\section{A Stabilizing Group for Maximally Entangled States}

Here we define the Weyl-Heisenberg operators, introduce the channels and show their equivalency. Then we discuss applications in quantum information theory.

\subsection{Weyl-Heisenberg and Weyl-Twirl Operators}
Let $\mathcal{H} = \mathcal{H}_d \otimes \mathcal{H}_d$ be the Hilbert space of two qudits of dimension $d$. In this bipartite system, we consider maximally entangled states, also called ``Bell states''. An orthonormal basis of $d^2$ Bell states spanning the Hilbert space can be defined by applying certain local operators on one of the subsystems to a seed state, e.g. the maximally entangled state $\ket{\Omega_{00}} := \frac{1}{\sqrt{d}} \sum_{i = 0}^{d-1} \ket{ii}$. Applying the Weyl-Heisenberg operators~\cite{teleportingWeyl}
\begin{gather}
    \label{weylOps}
    W_{k,l} := \sum_{j=0}^{d-1}w^{j k} \ket{j} \bra{j+l},~~k,l = 0,...,d-1
\end{gather}
with $w := e^{\frac{2 \pi i}{d}}$ to the (w.l.o.g.) first subsystem of $\ket{\Omega_{00}}$ defines the ``standard'' Bell basis of bipartite qudits:
\begin{gather}
  \label{bellStates}
  \ket{\Okl} := W_{k,l} \otimes \mathbb{1}_d \ket{\Omega_{00}},~~k,l = 0,...,d-1
\end{gather}
In eq.(\ref{weylOps}) and in the following, addition and subtraction are always to be understood$\mod d$. The indices $k$ and $l$ in $W_{k,l}$ relate to phase and shift operations, respectively. To highlight this property, $W_{k,l}$ can also be written as $W_{k,l} = Z(k)X(l)$, where $Z(k) = \sum_j w^{j \cdot k} \ketbra{j}{j}$ is the phase operator and $X(l) = \sum_j\ketbra{j-l}{j}$ denotes the shift operator. \\
The Weyl-Heisenberg operators obey the following algebraic relations:
\begin{equation}
    \label{weylRelations}
    \begin{aligned}
          & W_{k_1,l_1}W_{k_2,l_2} = w^{l_1 k_2}~W_{k_1+k_2, l_1+l_2}  \\
          & W_{k,l}^\dagger = w^{k l}~W_{-k, -l} = W_{k,l}^{-1} \\
          & W_{k,l}^\ast = W_{-k,l} \\
          & W_{k,l}^T = w^{-k l } W_{k, -l}.         
    \end{aligned}
\end{equation}
Here, $(\dagger), (\ast)$ and $(T)$ denote the adjoint, complex conjugation and transposition with respect to the computational basis, respectively. These relations imply a linear structure for the Bell states, which was recently shown~\cite{PoppQutritsAndQuquarts} to be highly relevant for the geometric properties of the set of separable, PPT entangled and NPT entangled states that are diagonal in the basis defined in eq.(\ref{bellStates}). \\
Consider now the set of $d^2$ unitary operators, which we call ``Weyl-Twirl operators'',
\begin{gather}
    \label{Tops}
    T_{i,j} := W_{i,j}\otimes W_{i,j}^\ast ,~~i,j = 0,...,d-1\;.
\end{gather}
Applying the Weyl relations (\ref{weylRelations}), one observes that this set of unitary operators $\lbrace \Tij|i,j = 0,..., d-1 \rbrace$ forms an abelian group under multiplication with neutral element $T_{0,0} = \mathbb{1}_d \otimes \mathbb{1}_d$:
\begin{equation}
    \label{groupT}
    \begin{aligned}
        & T_{i_1,j_1}T_{i_2,j_2} = T_{i_1+i_2, j_1+j_2} = T_{i_2,j_2}T_{i_1,j_1}\\
        & T_{i,j}^{-1} = T_{-i,-j}\;.
    \end{aligned}
\end{equation}
Using the well--known property of the maximally entangled state $\mathbb{1}_d \otimes M \ket{\Omega_{0,0}} = M^T \otimes \mathbb{1}_d \ket{\Omega_{0,0}} $ for any matrix $M$ together with the Weyl relations (\ref{weylRelations}), we can calculate the action of $\Tij$ on the Bell state $\ket{\Okl}$:
\begin{equation}
    \label{EigvalsT}
    \begin{aligned}
        \Tij \ket{\Okl} &= (W_{i,j} \otimes W_{i,j}^\ast) (W_{k,l} \otimes \mathbb{1}_d) \ket{\Omega_{0,0}}\\ 
        &= W_{i,j}W_{k,l}\otimes W_{-i,j} \ket{\Omega_{00}} \\
        &= W_{i,j}W_{k,l}W_{-i,j}^T \otimes \mathbb{1}_d \ket{\Omega_{00}} \\
        &= w^{jk-il} W_{k,l} \otimes \mathbb{1}_d \ket{\Omega_{00}} \\
        &= w^{jk-il} \ket{\Okl}
    \end{aligned}
\end{equation}
Each Bell state $\ket{\Okl}$ is thus an eigenvector for each $T_{i,j}$, with the phase of the eigenvalue being determined by the corresponding indices $(k,l)$ and $(i,j)$. 

\subsection{Pauli Projectors and the Weyl-Twirl Channel}

The properties of quantum channels, i.e. completely positive maps from one quantum state to another, lie at the heart of the dynamics of state propagation (see e.g. \cite{wilde_2017} for an overview). For instance, it is of interest to characterize channels which transform entangled states to separable states, so called entanglement-breaking channels (for a recent study see Ref.~\cite{Rudnicki}). Here we consider bipartite ``Pauli-type'' channels, i.e. channels whose Kraus operators are of the form $P_{k,l}$, that transform general states to Bell-diagonal ones.

Let us define the Bell projectors by  $\Pkl := \ket{\Okl}\bra{\Okl}$. The set $\M_d$ of Bell-diagonal states is given as mixtures of the projectors $\lbrace \Pkl \rbrace$ with mixing probabilities $\lbrace \ckl \rbrace$:
\begin{gather}
    \label{magicSimplex}
    \M_d := \lbrace \rho = \sumkl c_{k,l}\, P_{k,l}~ |~
    \sumkl c_{k,l} = 1, c_{k,l} \geq 0  \rbrace
\end{gather}
This object forms a mathematical simplex and is also known as the magic simplex~\cite{baumgartnerHiesmayr,baumgartnerHiesmayr2,baumgartner1}, where ``magic'' refers to the “magic'' Bell basis for bipartite qubits introduced by Wootters and Hill~\cite{HillWootters}.
Let $\rho = \sumkl \sumnm\rho_{k,l,m,n} \ket{\Okl}\bra{\Omega_{m,n}} \in \mathcal{H}$ be a bipartite state, represented in the Bell basis. We define the Pauli channel $\mathcal{P}: \mathcal{H} \rightarrow \M_d$ as a map from the total Hilbert space $\mathcal{H}$ to the set of Bell-diagonal states as follows:
\begin{equation}
    \label{pauliCh}
    \begin{aligned}
        \mathcal{P}(\rho) &:= \sumkl \Pkl\; \rho\; \Pkl = \sumkl \bra{\Okl} \rho \ket{\Okl}\, \Pkl \\
        &= \sumkl \rho_{k,l,k,l}\, \Pkl := \sumkl \ckl\, \Pkl
    \end{aligned}
\end{equation}
Using the operators $\Tij$ (\ref{Tops}), we can define another channel, which we name ``Weyl-Twirl'' channel:
\begin{gather}
    \label{twirl}
    \mathcal{T}(\rho) := \frac{1}{d^2}\; \sum_{i,j=0}^{d-1} \Tij\; \rho\; \Tij^\dagger = 
    \frac{1}{d^2}\; \sum_{i,j=0}^{d-1} \WWij\; \rho\; (\WWij)^\dagger
\end{gather}
We will show now that these two channels are equivalent, i.e. 
\begin{mytheorem}
    \label{theorem1}
\begin{equation}
    \label{chEquiv}
    \mathcal{P}(\rho) \equiv \mathcal{T}(\rho)~~\forall \rho \in \mathcal{H}
\end{equation}
\end{mytheorem}
\begin{proof}
Consider the action of $\mathcal{T}$ on a basis state $\ket{\Okl}\bra{\Omega_{m,n}}$:
\begin{equation*}
    \begin{aligned}
        \mathcal{T}(\ket{\Okl}\bra{\Omega_{m,n}}) &= \frac{1}{d^2}\sum_{i,j}^{d-1}  w^{j(k-m)-i(l-n)}\;\ket{\Okl}\bra{\Omega_{m,n}} \\
        &= \delta_{k,m} \delta_{l,n}\; \ket{\Okl}\bra{\Omega_{m,n}}
    \end{aligned}
\end{equation*}
Here, the first equality follows from (\ref{EigvalsT}) and the second equality from the identity $\sum_{j=0}^{d-1}w^{jx} = d \delta_{x,0}$. Let $\rho = \sum_{k,l,m,n=0}^{d-1}\rho_{k,l,m,n}\; \ket{\Okl}\bra{\Omega_{m,n}} \in \mathcal{H}$. The equation above implies that only diagonal elements remain, i.e.:
\begin{equation*}
    \begin{aligned}
        \mathcal{T}(\rho) = \sumkl \rho_{k,l,k,l} \ket{\Okl} \bra{\Okl} = \sumkl \rho_{k,l,k,l}\; \Pkl
        = \mathcal{P}(\rho)
    \end{aligned}
\end{equation*}
\end{proof}
Calling the application of the channel $\mathcal{T}$  ``Weyl-Twirl'' is motivated by the fact that eq.(\ref{twirl}) represents the random application of (bi-local) operators $\WWij$. Operations of the form $\rho \rightarrow \int (U \otimes U^\ast) \rho (U \otimes U^\ast)^\dagger dU $ with $U$ being a local unitary operator and $dU$ being the according Haar measure are generally named ``Twirl''. These operations leave certain diagonal elements invariant, while eliminating off-diagonal elements. Operations of this kind transform a state to Bell-diagonal form and are often relevant for certain applications, e.g. entanglement purification schemes~\cite{bennetEntPuri}. The channel equivalence of the Pauli channel (eq.\ref{pauliCh}) and the finite ``Weyl-twirl'' (eq.\ref{twirl})) shows that the operators $\Tij$ have the special property that they leave all Bell-diagonal elements invariant under random application, while eliminating any off diagonal elements.

\subsection{ Applications of Weyl-Twirl operators and channel}
The properties of the Weyl-Twirl operators $\Tij$ and the channel equivalence (eq.\ref{chEquiv}) imply several properties that are relevant for applications related to maximally entangled states. In the following, we will briefly mention two examples.

\subsubsection{The separability problem}
A bipartite quantum state $\rho \in \mathcal{H}_1 \otimes \mathcal{H}_2$ is separable ($\rho \in SEP$), if it can be written as convex combination of pure product states, that is $\rho = \sum_i q_i\rho_1^i \otimes \rho_2^i$ with $\rho_{1/2}^i \in \mathcal{H}_{1/2}$, $q_i>0$ and $\sum_i q_i = 1$.
Due to the existence of PPT-entangled states, the decision problem, whether a given bipartite state $\rho$ is separable ($\rho \in SEP$) or entangled ($\rho \notin SEP$), is generally a NP-hard problem~\cite{nphard,nphard-strong} called ``separability problem''. While all bipartite states with negative partial transposition (NPT) are known to be entangled according to the Peres-Horodecki/PPT-criterion~\cite{peres}, for subsystems of dimension $d \geq 3$, also entangled states with non-negative partial transposition (PPT) exist~\cite{distillation}. No general and efficient criterion is known that detects all PPT-entangled states nor a general method to construct those. Consequently, given a bipartite PPT state, it is generally unknown whether the state is entangled or not. However, for the set of Bell-diagonal states $\M_d$ \eqref{magicSimplex} in dimensions $d=3$ and $d=4$ the problem has efficiently be solved. In detail, for
$d=3$ ($4$), $95\%$ ($77\%$) of all PPT states have been classified as entangled or separable, using a combination of analytical criteria and numerical methods~\cite{PoppACS, PoppQutritsAndQuquarts}. Thus, for those dimension and the standard Bell basis the structure of PPT-entangled states is known.

As demonstrated in referenced works, the combination of two numerical methods are especially effective to distinguish separable and PPT-entangled Bell-diagonal states:
\begin{enumerate}
    \item The construction of separable states $\rho_s \in \M_d \cap SEP$, close to the border of the convex set of separable states, providing an inner approximating polytope of $\M_d \cap SEP$
    \item The construction of optimal entanglement witnesses (see definition below), providing an outer approximation of $\M_d \cap SEP$
\end{enumerate}
Both methods rely on an efficient parameterization of separable Bell-diagonal states $\rho_s \in \M_d \cap SEP$ that can be used to optimize corresponding target quantities. Given an efficient parameterization of unitaries and pure states~\cite{spenglerHuberHiesmayr, spenglerHuberHiesmayr2}, a parameterization of separable and Bell-diagonal states follows from the following corollaries of Theorem \ref{theorem1} :
\begin{mycor}
    \label{cor1}
    Any separable state $\rho_s$ remains separable under application of the Pauli channel $\mathcal{P}$, that is:
    $$
    \rho_s \in SEP \Rightarrow \mathcal{P}(\rho_s) \in \M_d \cap SEP
    $$
\end{mycor}
\begin{proof}
    Due to linearity, it suffices to show the corollary for a pure product state $\rho_s = \ket{\psi_1}\bra{\psi_1} \otimes \ket{\psi_2}\bra{\psi_2}$. By Theorem \ref{theorem1}, we have:
    \begin{equation*}
        \begin{aligned}
            \mathcal{P}(\rho_s) &= \mathcal{T}(\rho_s) \\
            &= \frac{1}{d^2}\sum_{i,j=0}^{d-1}  \WWij \;\rho_s\; (\WWij)^\dagger \\
            &= \frac{1}{d^2}\sum_{i,j=0}^{d-1}  W_{i,j} \ket{\psi_1}\bra{\psi_1} W_{i,j}^\dagger \otimes  W_{i,j}^\ast \ket{\psi_2}\bra{\psi_2} (W_{i,j}^\ast)^\dagger
        \end{aligned}
    \end{equation*}
    $W_{i,j}$ is unitary, so $W_{i,j} \ket{\psi_1}\bra{\psi_1} W_{i,j}^\dagger$ and $W_{i,j}^\ast \ket{\psi_2}\bra{\psi_2} (W_{i,j}^\ast)^\dagger$ are pure states. Therefore, the above expression represents an equal mixture of pure product states with mixing probability $\frac{1}{d^2}$ and thus is a separable mixed state. By definition, we also have $\mathcal{P}(\rho_s) \in \M_d$.     
\end{proof}
The second corollary states that optimal and Bell-diagonal entanglement witnesses are also optimal for $\M_d$. An entanglement witness $K$ is a Hermitian operator, for which the expectation $\Tr(K\rho_s)$ is non-negative for all separable states $\rho_s$, while at least one state $\rho_e$ exists with $\Tr(K \rho_e)<0$. In this case $K$ is said to detect the entangled state $\rho_e$. $K$ is said to be optimal, if there exists a separable state $\rho_0$, such that $\Tr(K\rho_0)=0$. It has been shown~\cite{baumgartnerHiesmayr} that Bell-diagonal entanglement witnesses can detect all entangled states in $\M_d$.
\begin{mycor}
    \label{cor2}
    If a Bell-diagonal entanglement witness $K$ is optimal, then it is also optimal for $\M_d$, i.e. $\exists \tilde{\rho_0} \in \M_d \cap SEP$ s.th. $\Tr(K \tilde{\rho_0}) = 0$.
\end{mycor}
\begin{proof}
    Let $\rho_0 = \sum_{k,l,m,n=0}^{d-1}\rho_{0(k,l,m,n)} \ket{\Omega_{k,l}}\bra{\Omega_{m,n}} \in SEP$ be a state so that $\Tr(K \rho_0) = 0$. Define $\tilde{\rho_0} := \mathcal{P}(\rho_0)$. By Corollary \ref{cor1}, $\tilde{\rho_0} \in \M_d \cap SEP$. $K$ is of Bell-diagonal form and Hermitian, so  $K=\sumkl \kappa_{k,l} \Pkl,~~ \kappa_{k,l} \in \mathbb{R}$. This implies:
    $$
        0 = \Tr(K \rho_0) = \sum_{k,l=0}^{d-1} \kappa_{k,l} \rho_{0(k,l,k,l)} = \Tr(K \tilde{\rho_0})
    $$
\end{proof}
Using the corollaries, a parameterization of $SEP$ together with convex optimization methods can then be applied to $\M_d \cap SEP$, by either optimizing for the convex set of $SEP$ and mapping the result to $\M_d$  (eq.(\ref{pauliCh})), or by directly considering the action of $\mathcal{P}$ for the parameterization.

\subsubsection{Error correction for a maximally entangled qudit}
The Weyl-Twirl operators $T_{i,j}$ (eq.\ref{Tops}) allow for a simple error identification and correction scheme for the process of sharing a maximally entangled state without access to the state itself. 
Assume that in some information processing task, the initial state $\ket{\Omega_{0,0}}$ is transformed to $\ket{\Okl}$ with probability $p_{k,l}$ so that the error channel state $\mathcal{E}(P_{0,0})$ is represented as
\begin{gather}
    \label{errorChannel}
    \mathcal{E}(P_{0,0}) = \sumkl p_{k,l}\Pkl\; .
\end{gather}
Another possibility that leads to the form of eq.\eqref{errorChannel} is a lost or unknown measurement outcome in the Bell basis.
The task is now to identify in which Bell states $\ket{\Okl}$ the system is without measuring or disturbing it and optionally to transform it back to the initial state $\ket{\Omega_{0,0}}$. The idea is to use eq.\eqref{EigvalsT} to store the state dependent phase $\Phi = jk-il$ in an ancilla qudit. After the phase has been identified, the state is known and can be transformed back to the initial state. \\
Consider a system to be in the initial state $\ket{\psi_0}$, where the Bell state $\ket{\Okl}$ is unknown and the ancilla qudit in a fixed state, i.e.
\begin{equation*}
    \begin{aligned}
        \ket{\psi_0} = \ket{0} \otimes \ket{\Okl}\;.
    \end{aligned}
\end{equation*}
Suppose further that the following operations are available:
\begin{itemize}
    \item Generalized Hadamard/Fourier gate $F$ acting as $F \ket{j} = \frac{1}{\sqrt{d}} \sum_{k=0}^{d-1} w^{kj} \ket{k}$    
    \item Controlled Weyl-Twirl gate $CT_{i,j}$ acting as $CT_{i,j}\; \ket{m} \otimes \ket{n} = \ket{m} \otimes T_{i,j}^m\; \ket{n}$
\end{itemize}
Consider the application of the Fourier gate to the ancilla, followed by the controlled Weyl-Twirl gate controlled by the ancilla with the unknown Bell state as target, followed by the adjoint Fourier gate on the ancilla qudit. These operations act as follows:
\begin{equation*}
    \begin{aligned}
       \ket{\psi_0} &\rightarrow (F^\dagger \otimes \mathbb{1})CT_{i,j}(F \otimes \mathbb{1}) (\ket{0} \otimes \ket{\Okl}) \\
       &= (F^\dagger \otimes \mathbb{1})CT_{i,j} (\frac{1}{\sqrt{d}} \sum_k \ket{k} \otimes \ket{\Okl}) \\
       &= (F^\dagger \otimes \mathbb{1}) (\frac{1}{\sqrt{d}} \sum_k \ket{k} \otimes T_{i,j}^k\ket{\Okl}) \\
       &= (F^\dagger \otimes \mathbb{1}) (\frac{1}{\sqrt{d}} \sum_k \ket{k} \otimes w^{k\Phi} \ket{\Okl}) \\
       &= (F^\dagger \frac{1}{\sqrt{d}} \sum_k w^{k\Phi} \ket{k}) \otimes \ket{\Okl} \\
       &= F^\dagger F\ket{\Phi} \otimes \ket{\Okl} = \ket{\Phi} \otimes \ket{\Okl}
    \end{aligned}
\end{equation*}
Measurement of the ancilla qudit yields the phase $\ket{\Phi} = \ket{jk-il}$, depending on the applied Weyl-Twirl operator through $i$ and $j$ and the unknown Bell state indexed by $k$ and $l$. Due to the stabilizing property of the operators $T_{i,j}$, the Bell state is not disturbed and the operation can be repeatedly applied or run in parallel with additional ancilla qudits. The applied Weyl-Twirl operators can be chosen according to the expected error. If only phase or only shift errors appear, a single measurement with suitably chosen $(i,j)$ can identify the error. If both shift and phase errors need to be identified, then two measurements of $\Phi$ obtained with different $(i,j)$ for  $T_{i,j}$ are required to identify the phase/shift error via $k$/$l$. The correction operation to recover $\ket{\Omega_{0,0}}$ is then $W_{k,l}^\dagger \otimes \mathbb{1}: \ket{\Okl} \rightarrow \ket{\Omega_{0,0}}$.

\section{Generalized Bell-diagonal systems}

Now we investigate cases, in which we still have a complete orthonormal Bell basis, but not via the standard construction introduced in the last section. In particular, this leads in general to a breaking of the channel equivalence $\mathcal{P}\equiv \mathcal{T}$.

\subsection{Generalized bases of Bell states}
Bell states are characterized by the fact that the reduced state for any of its subsystems is maximally mixed. Those states are called ``locally maximally mixed'' and they imply that all information about the state is in the correlations between the subsystems. The standard Bell basis defined in \eqref{bellStates} is of course an example of those states. However, there are bases that share this property, but are not unitarily equivalent to the standard basis\cite{baumgartnerHiesmayr}. \\
For $k,l=0,...,(d-1)$ consider the $d^2$ states 
\begin{equation}
    \label{altBellStatesPre}
    \begin{aligned}
      \ket{\phi^\alpha_{k,l}} := \frac{1}{\sqrt{d}} \sum_{s=0}^{d-1} w^{k(s-l)} \alpha_s\; \ket{s-l} \otimes \ket{s}
    \end{aligned}
\end{equation}
with $\alpha_s= e^{i\phi_s}$ being a phase factor. For a suitable choice of phase vectors $\alpha{\color{red}}=\{\alpha_s\}_{s=0}^{d-1}$, these states form an orthonormal set of locally maximally mixed basis states. Consider the corresponding projection operators
\begin{equation*}
    \begin{aligned}
      \ket{\phi^\alpha_{k,l}}\bra{\phi^\alpha_{k,l}} = \frac{1}{d} \sum_{s,t=0}^{d-1} w^{k(s-t)} \alpha_s \alpha_t^\ast\; \ket{s-l}\bra{t-l} \otimes \ket{s}\bra{t},
    \end{aligned}
\end{equation*}
for which any partial trace of the first or second subsystem yields the maximally mixed state:
\begin{equation}
    \begin{aligned}
     \tr_{1/2}(\ket{\phi^\alpha_{k,l}}\bra{\phi^\alpha_{k,l}}) = \frac{1}{d} \sum_{s=0}^{d-1} \alpha_s \alpha_s^\ast \ket{s}\bra{s} = \frac{1}{d} \sum_{s=0}^{d-1} \ket{s}\bra{s}
    \end{aligned}
\end{equation}
Furthermore, the orthonormality condition for $\ket{\phi^\beta_{m,n}}$ and  $\ket{\phi^\alpha_{k,l}}$ reads:
\begin{equation}
    \begin{aligned} 
    \delta_{m,k}\delta_{n,l} =
     \bra{\phi^\beta_{m,n}}\ket{\phi^\alpha_{k,l}} = 
     \frac{1}{d} \sum_{s,t}^{d-1} \alpha_s \beta_t^\ast\; w^{k(s-l)-m(t-n)} \;\bra{t-n} \ket{s-l} \bra{t} \ket{s} =
     \delta_{n,l} \frac{1}{d} \sum_{s} w^{s(k-m)}w^{mn-kl} \alpha_s \beta_s^\ast
    \end{aligned}
\end{equation}
This shows that for different shift indices ($n,l$), the Bell states \eqref{altBellStatesPre} are orthogonal, independent of the phases ($\alpha, \beta$). For equal shift indices, we now require $\alpha_s = \beta_s$, which implies:
\begin{equation}
    \begin{aligned}
     \bra{\phi^\beta_{m,l}}\ket{\phi^\alpha_{k,l}} = \frac{1}{d} w^{l(m-k)} \sum_{s} w^{s(k-m)} \alpha_s \alpha_s^\ast = \delta_{m,k}
    \end{aligned}
\end{equation}
Thus, in order to meet the requirements for orthonormality and being locally maximally mixed, we define the matrix $(\alpha_{s,t})_{s,t=0,...,d-1},~~|\alpha_{s,t}|=1~ \forall s,t$
and use it to define the ``generalized'' Bell basis $\lbrace \ket{\Phi_{k,l}^\alpha}|k,l = 0,...,d-1 \rbrace$ with
\begin{equation}
    \label{altBellStates}
    \begin{aligned}
      \ket{\Phi^\alpha_{k,l}} := \frac{1}{\sqrt{d}} \sum_s^{d-1} w^{k(s-l)} \alpha_{s,l}\; \ket{s-l} \otimes \ket{s}
    \end{aligned}
\end{equation}
In accordance with eq.\eqref{weylOps} and eq.\eqref{bellStates}, we define unitary operators that map the maximally entangled state $\ket{\Omega_{0,0}}$ to the generalized Bell basis states:
\begin{equation}
    \label{altWeylOp}
    \begin{aligned}
      &V^\alpha_{k,l} := \sum_j w^{j k} \alpha_{j+l,l}\; \ket{j}\bra{j+l},~~k,l=0,...,d-1 \\
      &\ket{\Phi^\alpha_{k,l}} = V^\alpha_{k,l} \otimes \mathbb{1} \;\ket{\Omega_{0,0}} \\
      & P_{k,l}^\alpha := \ket{\Phi^\alpha_{k,l}} \bra{\Phi^\alpha_{k,l}}\;.
    \end{aligned}
\end{equation}
Note that standard Weyl-Heisenberg operators $W_{k,l}$ and Bell basis states $\ket{\Okl}$ are a special case of eq.\eqref{altBellStates}, namely for $\alpha_{s,t} = 1 ~ \forall s,t$. In this sense we can talk about a generalization of the standard Weyl-Heisenberg operators and related Bell states.

\subsection{Properties of the generalized Bell bases}

The construction of generalized Bell basis introduced in the previous section (eq.\eqref{altWeylOp}) looks rather similar to the standard case (eqs.\eqref{weylOps}, \eqref{bellStates}) and generates also a set of orthonormal, maximally entangled and locally maximally mixed states. However, these Bell bases show significant differences with implications for the entanglement properties of diagonal states. We first state those differences and then numerically demonstrate their effect on the entanglement structure of corresponding systems of Bell-diagonal states. \\
The following differences to the standard construction hold for general transformation matrices $\alpha$ and can be checked by simple calculation or counter examples:
\begin{enumerate}
    \item The linear group structure does generally not exist: $V^\alpha_{k_1,l_1}V^\alpha_{k_2,l_2} \not\propto V^\alpha_{k_1+k_2, l_1+l_2}$
    \item The stabilizing property does generally not hold: $V^\alpha_{i,j} \otimes V_{i,j}^{\alpha\ast} \ket{\Phi_{k,l}} \neq w^\phi \ket{\Phi_{k,l}}$
    \item Let the simplex~\cite{baumgartnerHiesmayr} 
        \begin{gather}
            \label{altSimplex}
            \M^\alpha_d :=  \lbrace \rho = \sumkl c_{k,l}\; P^\alpha_{k,l}~ |~ \sumkl c_{k,l} = 1, c_{k,l} \geq 0  \rbrace
        \end{gather} be the set of states that are diagonal in the generalized Bell basis. The generalized Pauli channel 
        $$\mathcal{P}^\alpha: \mathcal{H} \rightarrow \M^\alpha_d,~~ \mathcal{P}^\alpha(\rho) := \sumkl P_{k,l}^\alpha \;\rho\; P_{k,l}^\alpha$$ 
        and the generalized Weyl-Twirl channel $$\mathcal{T}^\alpha(\rho) := \frac{1}{d^2}\sum_{i,j=0}^{d-1}  V^\alpha_{i,j}\otimes V^{\alpha\ast}_{i,j} \;\rho\; (V^\alpha_{i,j}\otimes V^{\alpha\ast}_{i,j})^\dagger$$ are generally not identical: $\mathcal{P}^\alpha \not\equiv \mathcal{T}^\alpha$
    \item  The generalized Pauli channel does generally not conserve separability, so Corollary \ref{cor1} does not hold, i.e.: 
    $\rho_s \in SEP \not\Rightarrow \mathcal{P}^\alpha(\rho_s) \in \M^\alpha_d \cap SEP $
    \item Due to the loss of the linear structure, also entanglement conserving symmetries of $\M_d$ (see Ref.~\cite{PoppQutritsAndQuquarts} and the references therein) are lost.
\end{enumerate}

\subsection{Explicit counter examples for bipartite qutrits}

\begin{figure}[h!]
    \centering
    \includegraphics[width=1.0\textwidth,keepaspectratio=true]{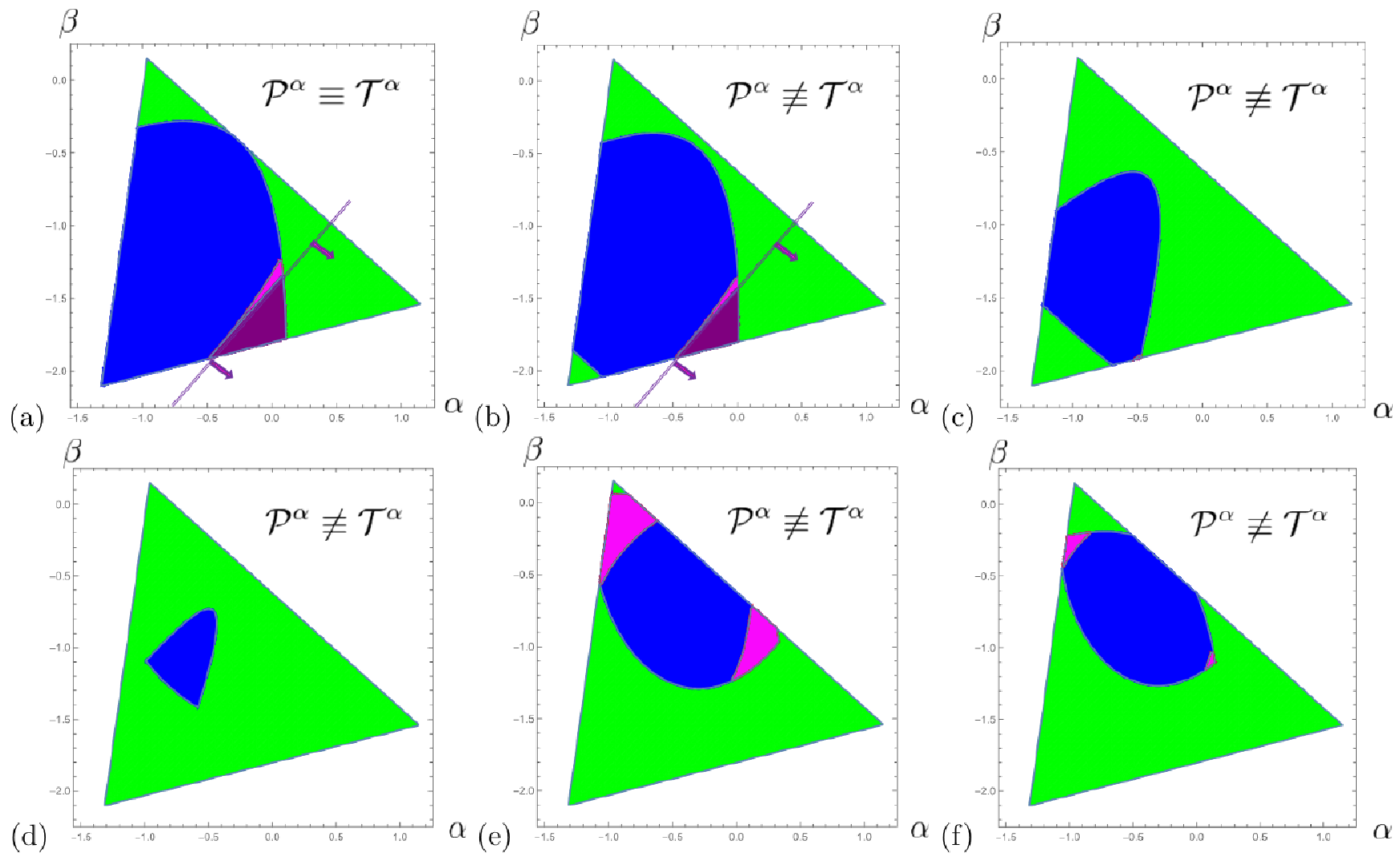}
    \caption{These pictures show a $2$-dimensional slice through the magic simplex for the state family, $\rho(\alpha,\beta)=(1-\frac{\alpha}{5}-\frac{\beta}{4}-\frac{1}{\sqrt{3}})\frac{1}{9}\mathbb{1}_9+\frac{\alpha}{5} P_{0,0}+\frac{\beta}{8} (P_{0,1}+P_{0,2})+\frac{1}{3 \sqrt{3}} (P_{1,0}+P_{1,1}+P_{1,2})$ which was investigated for the experimental proof of the existence of bound entanglement~\cite{hiesmayrLoeffler}. Each colored point represents one state parameterized by $(\alpha, \beta)$. If the color is green, then the corresponding state is NPT entangled, else PPT. The pink and purple regions show states where the realignment criterion (E2) and the witnesses used in the experiment~\cite{hiesmayrLoeffler} detect PPT-entangled states, respectively. Figure~(a) shows the result for states constructed with the standard basis choice, all others for non-standard bases. Figure~(b) shows the result for arbitrary but small phases. All other figures show the result for random phases in $[0,2\pi]$. We observe that the PPT-region (blue and pink/purple) changes dramatically and with that the possibility to find PPT-entangled states (pink/purple).}
    \label{fig1}
\end{figure}

In Fig.~\ref{fig1} we present a visualization of the drastic change in the separability structure of a state family within $\mathcal{M}_d^\alpha$ for dimension $d=3$ for non-standard Bell basis choices compared to the standard Bell basis choice.  The state family consists of mixtures of Bell projectors, which were experimentally investigated in Ref.~\cite{hiesmayrLoeffler}. This family is constructed to show the greatest violation for a specific entanglement witness based on mutually unbiased bases~\cite{ineqMUBs}, which can in principle be realized in experiments. Fig.~\ref{fig1}~(a) shows the result of the standard basis choice, including the region detected by the entanglement witness. In Fig.~\ref{fig1}~(b) the phases defining a generalized Bell basis (eq.(\ref{altWeylOp})) were randomly chosen but close to zero. One observes that the PPT-region shrinks as well as the region detected by the realignment criterion E2 (defined in the Appendix, named E2 according to Ref.~\cite{PoppACS}) and the entanglement witness. The other figures~(c)-(f) are based on random choices of the phases for defining the generalized Bell projectors and show how drastically the PPT region and the region detected by the realignment criterion (E2) changes. The Weyl-Heisenberg structure is thus very important for the separability structure within $\mathcal{M}_3^\alpha$.

\subsection{General entanglement structure changes for bipartite qutrits}

In the last section we have presented particular examples, but we can also come up with general results for $d=3$ by exploiting and adapting the tools developed in Refs.~\cite{PoppACS, PoppQutritsAndQuquarts}.
A combination of analytical and numerical methods~\cite{PoppJoss2023} exploiting the Heisenberg-Weyl structure were used to analyze the entanglement structure of Bell-diagonal qudits for $d=3$ and $d=4$, leading to an efficient solution of the separability problem in this particular case. As a consequence of the stated differences above, most of the tools cannot be applied for mixtures of the generalized Bell basis states \eqref{altBellStates}. However, some methods are still applicable that allow us to compare the standard basis with the generalized Bell bases for dimension $d=3$.\\
In order to compare the system $\M_3$ to $\M^\alpha_3$, we use the same uniformly distributed sample set of mixing probabilities $\ckl$ to construct $10,000$ states in $\M_3$ and in $\M^\alpha_3$ for various $\alpha$. In particular, we generate $1,000$ matrices $\alpha$ with uniformly distributed elements. For each $\alpha$, we define the Bell basis projectors $\Pkl^\alpha$ and construct $10,000$ diagonal states according to \eqref{altSimplex}. We then analyze these states for PPT and entanglement detection by the realignment criterion~\cite{realignment} (named E2 in Ref.~\cite{PoppACS}) and the quasipure concurrence criterion~\cite{BaeQuasipure} (named E3 in Ref.~\cite{PoppACS}) (see Appendix for definitions), which do not depend on the entanglement class preserving symmetries of the standard simplex $\M_3$, and are the most successful analytical criteria among those investigated in Ref.~\cite{PoppACS, PoppQutritsAndQuquarts}.\\
\begin{table}
\parbox{.50\linewidth}{
    \begin{center}
    \begin{tabular}{ l|c|c|c|c } 
    & Std. Basis & Min & Max & Mean \\
     \hline
        rPPT &  $60.0\%$ & $49.0\%$ & $59.0\%$ & $51.6\%$ \\ 
         \hline    
        E2/PPT & $10.4\%$ & $2.1\%$ & $9.6\%$ & $3.6\%$ \\ 
         \hline
        E3/PPT & $2.7\%$ & $0.6\%$ & $2.3\%$ & $1.0\%$ \\ 
        \hline
        (E2\&E3)/PPT & $1.8\%$ & $0.1\%$ & $1.4\%$ & $0.4\%$ \\
    \end{tabular}
    \caption{Minimum, maximum and mean statistics for 1,000 sample systems $\M_3^\alpha$ and the reference values for the standard system $\M_3$. Relative volume of PPT states (rPPT) and the share of PPT-entangled states among them for the realignment criterion (E2/PPT), the quasipure approximation criterion (E3/PPT) and the combined criterion ((E2\&E3)/PPT).}
    \label{detectionVolumes}
    \end{center}
}
\hfill
\parbox{.45\linewidth}{
    \begin{center}    
    \begin{tabular}{ l|c } 
        & rPPT   \\
         \hline
        E2/PPT &  $0.99$   \\ 
         \hline    
        E3/PPT & $0.95$ \\ 
        \hline
        (E2\&E3)/PPT & $0.96$ \\
    \end{tabular}
    \caption{Correlation coefficient for the relative volume of PPT states (rPPT) and the share of PPT-entangled states among the PPT states as detected by the realignment criterion (E2/PPT), the quasipure concurrence criterion (E3/PPT) and the combined criterion ((E2\&E3)/PPT).}
    \label{correlations}
    \end{center}
}
\end{table}   
Table \ref{detectionVolumes} shows the share of PPT states among all analyzed states (rPPT), the share of PPT-entangled states detected by the realignment criterion among all PPT states (E2/PPT), and the corresponding share for the quasipure concurrence criterion (E3/PPT), as well as the share of entangled states that were detected by both criteria simultaneously (E2\&E3)/PPT. It shows that for the 1,000 realizations of $\M_3^\alpha$ the relative volume of PPT states is between $49\%$ and $59\%$. The mean of $51.6\%$ indicates most of the systems are closer to the minimum value than to the maximum. Indeed, each of more than $85\%$ of the systems have less than $54\%$ PPT states. Curiously, the PPT share for the standard system $\M_3$ is higher than for any analyzed $\M_3^\alpha$. For values of all $\alpha$ being sufficiently close to $1$, i.e. for Bell basis being close to the standard Bell basis, the observed quantities are arbitrarily close, but never higher then the reference values for the standard system. The same statements hold for the detection capabilities of E2, E3 and (E2\&E3), which are always lower than the value for the standard system and are typically closer to the minimal observed value.\\
Table \ref{correlations} presents strong positive correlations between the share of PPT states in the systems $\M_3^\alpha$ and the relative detection capabilities of the criteria E2, E3 and simultaneous detection (E2\&E3). Interestingly, the more PPT states are present, the larger the relative volume of PPT-entangled states is, which are detected by those criteria. Moreover, the relative number of simultaneously detected states is then also high, in general. The strong correlations imply an almost linear dependence between those quantities. \\
Based on these observations, one can assume that the more states with positive partial transposition are present in a Bell-diagonal system, the larger the share of PPT-entangled states, detected by the two effective entanglement criteria E2 and E3, is. Additionally, a higher share of PPT states seems to imply a higher share of PPT-entangled states that can be detected by both of the criteria simultaneously. While in the standard system only $33\%$ of the states detected by E3 are not also detected by E2, there are Bell-diagonal systems with less PPT states, in which $83\%$ of the entangled states detected by E3 are not detected by E2. Curiously, the standard system shows extreme values for the relative volume of PPT states and the number of entangled states among them.

\section{Summary and outlook}
In this work, we investigated and compared properties of bipartite maximally entangled and locally maximally mixed Bell states, known as Bell-diagonal states and showed that the special features of the Weyl-Heisenberg Bell basis imply special features in the entanglement structure of Bell-diagonal states.

The frequently used ``standard'' construction of a $d^2$ dimensional Bell basis of the joint Hilbert space $\mathcal{H}_d \otimes \mathcal{H}_d$ via the Weyl-Heisenberg operators is presented and properties of this special basis were derived that are strongly related to the separability problem and other applications, for example, error corrections or channel equivalences. In particular, utilizing the Weyl relations of the Weyl-Heisenberg operators $W_{k,l}$, we showed that the ``Weyl-Twirl'' operators $W_{i,j} \otimes  W_{i,j}^\ast$ are diagonalized by all elements of the standard Bell basis. We then leveraged this stabilizing property to show the equivalence of the ``Pauli'' channel, $\mathcal{P}$, which projects onto Bell-diagonal states in the standard Bell basis, and the ``Weyl-Twirl'' channel, $\mathcal{T}$, which represents the randomized application of the Weyl-Twirl operators. 

One implication of this channel equivalence, $\mathcal{P}\equiv\mathcal{T}$, is that the separability is conserved under the channels and that optimal Bell-diagonal entanglement witnesses remain optimal for the set of Bell-diagonal states. We then demonstrated several applications of the Weyl-Twirl operators. On the one hand, they allow the efficient parameterization of separable Bell-diagonal states, relevant for polytope approximations of this convex set and the construction of optimal Bell-diagonal entanglement witnesses. These applications, enabled by the special structure of investigated Bell states, allow the effective entanglement classification and detection of PPT-entanglement for Bell-diagonal states.
On the other hand, a simple error detection and correction scheme for a maximally entangled qudit was presented.

Highlighting the implication of the channel equivalence, $\mathcal{T}\equiv\mathcal{P}$, we investigated systems of Bell-diagonal states that are constructed from generalized Bell basis states. Those generalized Bell states form an orthonormal basis and are locally maximally mixed, similar to the standard basis based on Heisenberg-Weyl operators. However, the channel equivalence $\mathcal{T}\equiv\mathcal{P}$ and other properties are lost, which has strong implications on the entanglement structure within the family of Bell-diagonal states. We summarize the most relevant objects, their properties and relations in Fig.\ref{fig:2} below. Thereby, we compare the results related to the standard Bell basis to the generalized Bell bases.

In detail, we analyzed this for bipartite qutrits. We derived the relative volume of PPT states among Bell-diagonal states. For the standard basis, the relative volume is $60\%$ while it was shown to drop to $49\%$ for some generalized Bell bases. Interestingly, we found that the share of Bell-diagonal PPT states is generally lower for the generalized Bell bases compared to the standard Bell basis. Moreover, the relative detection rate for PPT-entangled states of two, in the standard system highly effective, entanglement criteria strongly correlates to the share of PPT states. The more states with positive partial transposition exist for a system of Bell-diagonal states, the higher is the share of detected PPT-entangled states among them for those two criteria. Consequently, for the system based on the standard Bell basis, which has the highest PPT share, also the highest relative amount of PPT-entanglement is observed. Furthermore, we visualized the dramatic change of the volume of PPT states in Fig.~\ref{fig1} for a family of states investigated previously in experiment. 

These results indicate that, among all in this way generalized Bell bases, Bell-diagonal states related to the standard construction have some very special properties. These properties have relevant implications on the entanglement structure and on practical applications using those Bell states. 
Numerical results suggest that either the two entanglement criteria are less effective for Bell-diagonal states if the special properties of the standard system are not given or that the structure and amount of PPT entanglement depend on those properties. 

In summary, our findings are the starting point to analyze different quantum information theoretic protocols with regards to their dependence on the underlying structure and properties of used Bell states.

\begin{figure}[H]
\label{overviewdiagram}
\centering

\tikzstyle{block} = [draw,rectangle,thick,minimum height=2em,minimum width=2em, text width=3.2cm, text centered] 
\tikzstyle{line} = [draw, very thick, color=black!50, -latex']
\begin{tikzpicture}[node distance = 1cm, auto] 
    \node [block] (weylops) {Weyl operators: \\$W_{k,l}$ }; 
    \node [block, left of=weylops, xshift=-3cm, yshift=-1cm] (weylrel) {Weyl relations: \\Yes};
    \node [block, below of=weylops, yshift=-2cm] (weylbasis) {Bell basis: \\ $\ket{\Okl}$};
    \node [block, text width=4cm, left of=weylbasis, xshift=-3cm] (twirl) {Weyl-Twirl: \\$T_{i,j} = W_{i,j}\otimes W_{i,j}^\star$ \\
    $T_{i,j} \ket{\Okl} = w^{jk-il} \ket{\Okl}$
    }; 
    \node [block, below of=weylbasis, yshift=-1cm] (bds) {Diagonal states: \\$\M_d$}; 
    \node [block, left of=bds, xshift=-3cm] (channels) {Channels:\\ $\mathcal{P} \equiv \mathcal{T}$}; 
    \node [block, below of=bds, yshift=-1cm] (ppt) {Relative amount PPT in $\M_3$:\\ $60\%$}; 
    \node [block, below of=ppt, yshift=-1cm] (bound) {Detected PPT-entanglement: \\$11.3\%$}; 

    \path [line] (weylops) -- (weylbasis) ; 
    \path [line] (weylops) -- (weylrel) ;
    \path [line] (weylops) -- (twirl) ; 
    \path [line] (weylbasis) -- (bds) ; 
    \path [line] (twirl) -- (channels) ; 
    \path [line] (bds) -- (ppt) ; 
    \path [line] (ppt) -- (bound) ; 

    \node [block, right of=weylops, xshift=3.2cm] (genweylops) {gen. Weyl operators: \\$V^\alpha_{k,l}$ }; 
    \node [block, right of=genweylops, xshift=3cm, yshift=-1cm] (genweylrel) {Weyl relations:\\No};
    \node [block, below of=genweylops, yshift=-2cm] (genweylbasis) {gen. Bell basis: \\$\ket{\Phi^\alpha_{k,l}}$};
    \node [block, below of=genweylbasis, yshift=-1cm] (genbds) {Diagonal states:\\ $\M^\alpha_d$}; 
    \node [block, text width=4cm, right of=genweylbasis, xshift=3cm] (gentwirl) {gen. Weyl-Twirl:\\ $T^\alpha_{i,j} = V^\alpha_{i,j}\otimes V^{\alpha\star}_{i,j}$\\
    $T^\alpha_{i,j} \ket{\Phi^\alpha_{k,l}} \neq w^\phi \ket{\Phi^\alpha_{k,l}}$}; 
    \node [block, right of=genbds, xshift=3cm] (genchannels) {Channels:\\ $\mathcal{P^\alpha} \not\equiv \mathcal{T^\alpha}$}; 
    \node [block, below of=genbds, yshift=-1cm] (genppt) {Relative amount PPT in $\M_3^\alpha$:\\ $\leq 60\%$}; 
    \node [block, below of=genppt, yshift=-1cm] (genbound) {Detected PPT-entanglement:\\ $\leq 11.3\%$}; 

    \path [line] (genweylops) -- (genweylbasis) ; 
    \path [line] (genweylops) -- (genweylrel) ;
    \path [line] (genweylops) -- (gentwirl) ; 
    \path [line] (genweylbasis) -- (genbds) ; 
    \path [line] (gentwirl) -- (genchannels) ; 
    \path [line] (genbds) -- (genppt) ; 
    \path [line] (genppt) -- (genbound) ; 
\end{tikzpicture}
\caption{Comparison of relevant objects and their relation for the standard and generalized Weyl-construction. Differences include the lack of linear structure (Weyl relations) of generalized Weyl operators, the lack of the stabilizing property of the Weyl-Twirl operators and the general inequivalence of Pauli channel and Weyl-Twirl channel. In the generalized case, the amount of PPT states and the relative amount of detected PPT-entanglement is lower compared to the standard system.} 
\label{fig:2}
\end{figure}
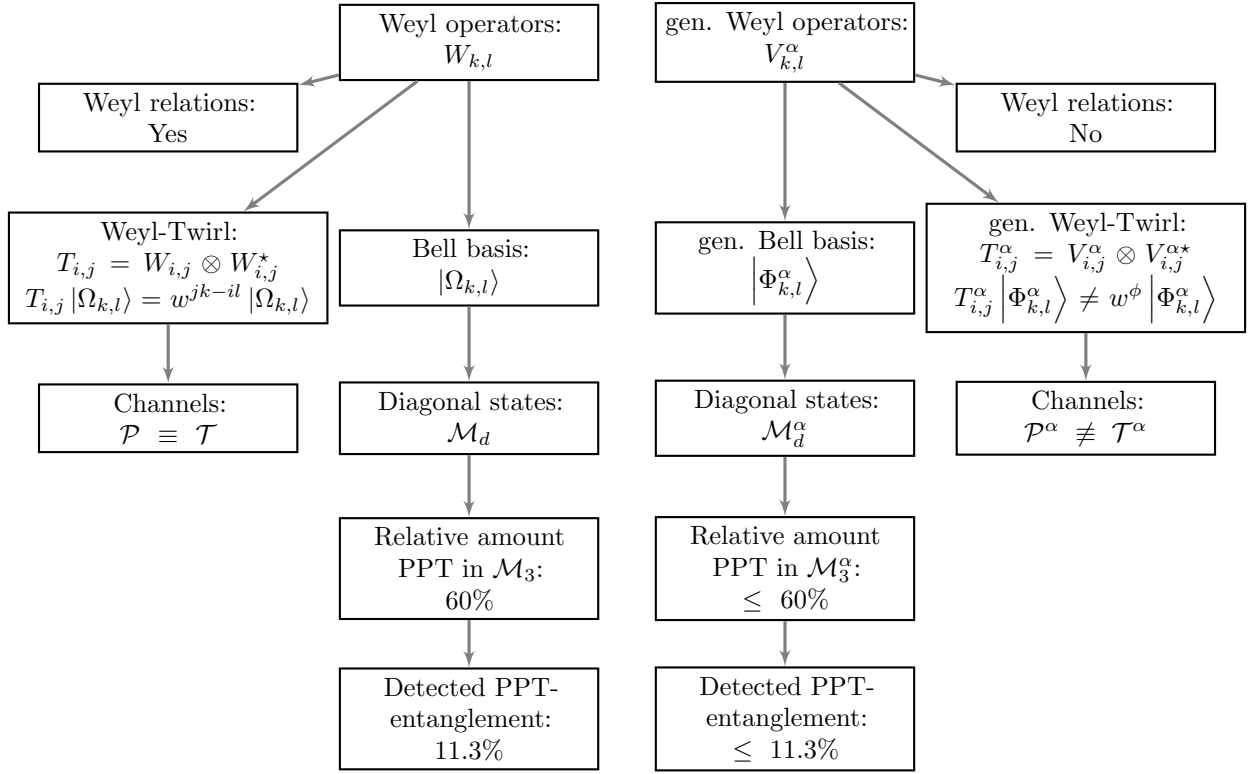

\section*{Data availability statement}
All analyzed datasets were generated during the current study and are available
from the corresponding author on reasonable request. \\

\section*{Code availability statement}
The software used to generate the reported results is published as repository and registered open source package ``BellDiagonalQudits.jl'' ~\cite{PoppJoss2023} available at \url{https://github.com/kungfugo/BellDiagonalQudits.jl}.

\printbibliography 

\section*{Acknowledgments}
B.C.H. and C.P. acknowledge gratefully that this research was funded in whole, or in part, by the  Austrian Science Fund (FWF) project P36102-N. For the purpose of open access, the author has applied a CC BY public copyright licence to any Author Accepted Manuscript version arising from this submission. 

\section*{Author contributions statement}
C.P. carried out the analytical and numerical analyses, developed the new methods and implemented the software.\\
B.C.H. revised the analyses and added results and proposed improvements. \\
C.P and B.C.H. both edited the manuscript.

\section*{Competing interests}
 All authors declare no financial or non-financial competing interests.
 
\section*{Additional information}
Correspondence and requests for materials should be addressed to C.P..

\newpage

\section*{Appendix}
\subsection*{Applied entanglement criteria}
\textbf {Peres-Horodecki/PPT criterion} \\
The partial transpose $\Gamma$
acts on the computational basis states of a bipartite state as
$(\ketbra{i}{j} \otimes \ketbra{k}{l})^{\Gamma} := \ketbra{i}{j} \otimes
\ketbra{l}{k}$. If the partial transpose for a density matrix of a bipartite state has at least one negative eigenvalue (in which case it is said to be ``NPT''), it is entangled ~\cite{peres}. For $d=2$ it
detects all entangled states, but for $d \geq 3$ it is only sufficient due to
the existence of PPT-entangled states. \\ \\
\textbf{E2: Realignment criterion} \\
The realignment operation $R$ is defined for the computational basis states as
$(\ketbra{i}{j} \otimes \ketbra{k}{l})_R := \ketbra{i}{k} \otimes
\ketbra{j}{l}$. If the
sum of singular values of the density matrix of a realigned state $\sigma_R$ are larger than $1$,
then $\sigma$ is entangled ~\cite{realignment}. \\ \\
\textbf{E3: Quasipure concurrence criterion} \\
The quasipure approximation~\cite{BaeQuasipure} $C_{qp}$ provides an efficiently computable lower bound on the concurrence~\cite{woottersConcurrence} $C$. Let $i$ and $j$ ($k$ and $l$) enumerate the computational basis vectors of the first (second) partition of the Hilbert space. Define
$A := 4 \sum_{i<j, k<l} \ket{ikjl}-\ket{jkil}-\ket{iljk}+\ket{jlik} \times h.c.$ and let a state be represented in its spectral composition $\rho = \sum_i \mu_i \ket{\Psi_i}\bra{\Psi_i}$. With a dominant eigenvector $\ket{\Psi_0}$, define $\ket{\xi} \propto A \ket{\Psi_0} \otimes \ket{\Psi_0}$ and 
$$T_{i,j} := \sqrt{\mu_i \mu_j} \bra{\Psi_i} \otimes \bra{\Psi_j} \ket{\xi}.$$ The quasipure concurrence is defined by the singular values $S_i$ of the matrix $T_{i,j}$ 
$$
C_{qp}(\rho) := \max(0, S_0 - \sum_{i>0}S_i) \leq C(\rho)
$$ and detects entanglement for positive values.\\
For the standard system $\M_d$, the approximation takes an explicit form: Here, the singular values $S_{k,l}$ are explicitly given by
\begin{gather*}
S_{k,l} = \sqrt{
  \frac{d}{2(d-1)} c_{k,l}
  [(1-\frac{2}{d}) c_{n,m} \delta_{k,n} \delta_{l,m}
  + \frac{1}{d^2} c_{(2n-k)mod~d,(2m-l)mod~d}]
}
\end{gather*}
where $(n,m)$ is a multi-index of the coordinate of the largest value $\lbrace c_{k,l} \rbrace $~\cite{BaeQuasipure}.

\end{document}